%% file: counterex.tex
\documentclass[11pt,a4paper]{article}

\usepackage[margin=2.50cm]{geometry}

\usepackage{url}
\pdfoutput=1

\input{environmentsRD}
\input{gencommands}

\usepackage{enumitem}
\usepackage[pdftex,pagebackref=true]{hyperref}

\title{An obstruction to Delaunay triangulations in\\
Riemannian manifolds
}

\author{
Jean-Daniel Boissonnat
\footnote{
INRIA, DataShape,
Sophia-Antipolis, France
\url{Jean-Daniel.Boissonnat@inria.fr}
}
\and
Ramsay Dyer
\footnote{
INRIA, DataShape
\url{Ramsay.Dyer@inria.fr}
}
\and
Arijit Ghosh
\footnote{
Indian Statistical Institute,
ACM Unit,
Kolkata, India
\url{arijitiitkgpster@gmail.com}
}
\and
Nikolay Martynchuk
\footnote{
Johann Bernoulli Institute for Mathematics and Computer Science, 
University of Groningen, 
P.O. Box 407, 9700 AK, 
Groningen, The Netherlands
\url{N.Martynchuk@rug.nl}
}
}

\begin{document}

\maketitle

\begin{abstract}
  Delaunay has shown that the Delaunay complex of a finite set of
  points $\mpts$ of Euclidean space $\rem$ triangulates the convex
  hull of $\mpts$, provided that $\mpts$ satisfies a mild genericity
  property. Voronoi diagrams and Delaunay complexes can be defined for
  arbitrary Riemannian manifolds. However, Delaunay's genericity
  assumption no longer guarantees that the Delaunay complex will yield
  a triangulation; stronger assumptions on $\mpts$ are required. A
  natural one is to assume that $\mpts$ is sufficiently dense.
  Although results in this direction have been claimed, we show that
  sample density alone is insufficient to ensure that the Delaunay
  complex triangulates a manifold of dimension greater than 2.
\end{abstract}

\input{intro}
\input{scale_invariant}
\input{discussion}

\subsection*{Acknowledgements}

We thank Gert Vegter for suggesting that we look for an obstruction
that can exist at all scales.  We also thank David Cohen-Steiner and
Mathijs Wintraecken for illuminating discussions.

This research has been partially supported by the 7th Framework
Programme for Research of the European Commission, under FET-Open
grant number 255827 (CGL Computational Geometry Learning). Partial 
support has also been provided by the Advanced Grant of the European 
Research Council GUDHI (Geometric Understanding in Higher Dimensions).

Arijit Ghosh is supported by Ramanujan Fellowship
number SB/S2/RJN-064/2015. Part of this work was done when he
was a researcher at the Max-Planck-Institute for Informatics, Germany
supported by the IndoGerman Max Planck Center for Computer Science
(IMPECS). Part of this work was also done while he was a
visiting scientist at the Advanced Computing and Microelectronics Unit,
Indian Statistical Institute, Kolkata, India.

\bibliographystyle{alpha}
\bibliography{delrefs}

\end{document}

%% file: environmentsRD.tex
\usepackage{amsmath,amssymb,euscript}

\usepackage{graphicx}
\usepackage{color}

\usepackage{ntheorem}

\usepackage{float}

\floatstyle{ruled}
\newfloat{structure}{Hhtb}{los}
\floatname{structure}{Structure}

\usepackage{subfig}
\usepackage{algorithmic}
\usepackage{algorithm}
\usepackage{wrapfig}

\newcounter{prob}
        {\end{list}}
\newenvironment{proof}[1][Proof.]{
  \begin{trivlist}\item[]\textit{#1}\, }%
  {\hfill\hspace*{\fill}~$\square$\end{trivlist}}

\theorembodyfont{\normalfont}
\newtheorem{theorem}{Theorem}
\newtheorem{definition}{Definition}


%% file: gencommands.tex
\definecolor{turquoise}{cmyk}{0.65,0,0.1,0.1}

\newcommand{\rawdef}[1]{\emph{#1}} 
\newcommand{\defn}[1]{\rawdef{#1}\index{#1}}

\newcommand{\Defref}[1]{Definition~\ref{#1}}
\newcommand{\Eqnref}[1]{Equation~\eqref{#1}}
\newcommand{\Figref}[1]{Figure~\ref{#1}}

\newcommand{\Secref}[1]{Section~\ref{#1}}

\newcommand{\Thmref}[1]{Theorem~\ref{#1}}

\newcommand{\cinfty}{C^\infty}

\newcommand{\R}{\mathbb{R}}
\newcommand{\reel}{\mathbb{R}}

\newcommand{\rthree}{\reel^3}

\newcommand{\rem}{\reel^m}

\newcommand{\ints}{\mathbb{Z}}

\newcommand{\norm}[1]{\left\|#1\right\|}
\newcommand{\abs}[1]{\left|#1\right|}

\newcommand{\size}[1]{\left|#1\right|}

\newcommand{\bdry}[1]{\partial{#1}}

\newcommand{\asimplex}[1]{\{#1\}} 
\newcommand{\carrier}[1]{\abs{#1}} 

\newcommand{\gdist}{d} 

\newcommand{\gdistG}[1]{\gdist_{#1}} 
\newcommand{\distG}[3]{\gdist_{#1}(#2,#3)}
\newcommand{\distE}[2]{\distG{\rthree}{#1}{#2}}

\newcommand{\gdistM}{\gdistG{\man}}
\newcommand{\distM}[2]{\distG{\man}{#1}{#2}}

\newcommand{\spaceball}[3]{B_{#1}(#2,#3)} 

\newcommand{\ballM}[2]{\spaceball{\man}{#1}{#2}}

\DeclareMathOperator{\aff}{aff} 
\newcommand{\affhull}[1]{\aff(#1)}

\newcommand{\mpts}{P} 

\newcommand{\man}{\mathcal{M}}

\DeclareMathOperator{\Vor}{Vor}

\newcommand{\vorMmpts}{\Vor_{\man}(\mpts)}
\newcommand{\vorcell}[2]{\mathcal{V}_{#1}(#2)}

\newcommand{\vorcellman}[1]{\vorcell{\man}{#1}}

\DeclareMathOperator{\Del}{Del}
\newcommand{\del}[2]{\Del_{#1}(#2)} 

\newcommand{\delM}[1]{\del{\man}{#1}} 
\newcommand{\delMmpts}{\delM{\mpts}}

\newcommand{\samconst}{\epsilon}
\newcommand{\tsamconst}{\tilde{\epsilon}}
 
\newcommand{\splxs}{\sigma}

\newcommand{\splxt}{\tau}

\DeclareMathOperator{\injr}{inj}
\newcommand{\injrad}[1]{\injr(#1)}
\newcommand{\injradM}{\injrad{\man}}

\newcommand{\bmax}{b_{\mathrm{max}}}
\newcommand{\txi}{\tilde{\xi}}

\def\unfrac#1#2{#1/#2}


%% file: intro.tex
%

\section{Delaunay complex and Delaunay triangulation}

Let $(\man, \gdistM)$ be a metric space, and let $\mpts$ be a
finite set of points of $\man$.  An \defn{empty ball} is an open ball
in the metric $\gdistM$ that contains no point from $\mpts$.  We say
that an empty ball $B$ is \defn{maximal} if no other empty ball with
the same centre properly contains $B$.  A \defn{Delaunay ball} is a
maximal empty ball.

A simplex $\splxs$ is a \defn{Delaunay simplex} if there exists some
Delaunay ball $B$ that circumscribes $\splxs$, i.e., such that the
vertices of $\splxs$ belong to $\bdry{B}\cap \mpts$.  The
\defn{Delaunay complex} is the set of Delaunay simplices, and is
denoted $\delMmpts$. It is an abstract simplicial complex and so
defines a topological space, $\carrier{\delMmpts}$, called its
\defn{carrier}. We say that $\delMmpts$ \defn{triangulates} $\man$ if
$\carrier{\delMmpts}$ is homeomorphic to $\man$.  A Delaunay
\defn{triangulation} of $\man$ is a homeomorphism
$H\colon \carrier{\delMmpts} \to \man$.

The \defn{Voronoi cell} associated with $p \in \mpts$
is given by
\begin{equation*}
  \vorcellman{p} = \{x \in \man \mid \distM{x}{p} \leq \distM{x}{q}
  \text{ for all } q \in \mpts \}.
\end{equation*}
More generally, a \defn{Voronoi face} is the intersection of a set of
Voronoi cells: given $\splxs = \{p_0,\ldots, p_k\} \subset \mpts$, we
define the associated Voronoi face as
\begin{equation*}
  \vorcellman{\splxs} = \bigcap_{i=0}^{k} \vorcellman{p_i}.
\end{equation*}
It follows that $\splxs$ is a Delaunay simplex if and only if
$\vorcellman{\splxs} \neq \emptyset$. In this case, every point in
$\vorcellman{\splxs}$ is the centre of a Delaunay ball for $\splxs$.
Thus every Voronoi face corresponds to a Delaunay simplex. The Voronoi
cells give a decomposition of $\man$, denoted $\vorMmpts$, called the
\defn{Voronoi diagram}.  The Delaunay complex of $\mpts$ is the nerve
of the Voronoi diagram.

In the case of $\rem $ equipped with the standard Euclidean metric,
Delaunay~\cite{delaunay1934} showed that, if $\mpts$ is
\defn{generic}, then the natural inclusion
$\mpts \hookrightarrow \rem$ induces a piecewise linear embedding
$\carrier{\del{\rem}{\mpts}} \hookrightarrow \rem$ of the Delaunay
complex of $\mpts$ into $\rem$.
This is
called the Delaunay triangulation of $\mpts$ (it is a triangulation of
the convex hull of $\mpts$).
The point set $\mpts$ is generic if there is no Delaunay ball with
more than $m+1$ points of $\mpts$ on its boundary.  Point sets that
are not generic are often dismissed in theoretical work, because
an
arbitrarily small perturbation of the points will almost surely
yield a generic point set \cite{delaunay1934}. 
In the sense of the standard measure
in the configuration space $\reel^{m \times \size{\mpts}}$, almost all
point sets will yield a Delaunay triangulation.

A similar situation is known for certain standard non-Euclidean
geometries, such as Laguerre geometry~\cite{Aurenhammer87}, spaces
equipped with a Bregman divergence~\cite{BoissonnatNN10-Bregman}, or
Riemannian manifolds of constant sectional curvature\footnote{This is
  standard folklore. We are not aware of a published reference.}.

Leibon and Letscher~\cite{leibon2000} announced sampling density
conditions which would ensure that the Delaunay complex defined by the
intrinsic metric of an arbitrary compact Riemannian manifold admits a
triangulation.  When triangulating submanifolds of dimension $3$ and
higher in Euclidean space using Delaunay techniques, it was
subsequently discovered that near degenerate ``sliver'' simplices pose
problems which cannot be escaped simply by increasing the sampling
density.  In particular, developing an example on a $3$-manifold
presented by Cheng et al.~\cite{cheng2005}, Boissonnat et
al.~\cite[Lemma 3.1]{boissonnat2009} show that, using the metric of
the ambient Euclidean space restricted to the submanifold, the
resulting Delaunay complex (called the \defn{restricted Delaunay
  complex}) need not triangulate the original submanifold, even with
dense well-separated (no two points are too close) sampling.

Here we develop a similar example from the perspective of the
intrinsic metric of the manifold. It can be argued that this is an
easier way to visualize the problem, since we confine our viewpoint to
a three dimensional space and perturb the metric, without referring to
deformations into a fourth ambient dimension. This viewpoint also
provides an explicit counterexample to the results announced by Leibon
and Letscher~\cite{leibon2000}.  We construct a fixed compact
Riemannian manifold and demonstrate (\Thmref{thm:counter.ex}) that for
any sampling density there exists a point set that meets the sampling
density and has good separation, but that does not admit a Delaunay
triangulation, and furthermore, this property is retained when the
point set is subjected to a sufficiently small perturbation. Thus, not
only is there no sampling density that is sufficient to ensure the
existence of a Delaunay triangulation on all compact Riemannian
manifolds, but for any given fixed compact manifold, there may exist
no sampling density that can guarantee a Delaunay triangulation. In
particular, this latter property will complicate attempts to use
Delaunay techniques in the construction and asymptotic analysis of
optimal triangulations \cite{clarkson2006,deLaat2011}.

Although we focus on the intrinsic metric, the same qualitative
construction applies to the restricted ambient metric, demonstrating
that the claim by Cairns~\cite{cairns1961}, that the Voronoi faces of
the restricted Voronoi diagram are closed topological balls if the
sampling is sufficiently dense, is also incorrect. 

The main purpose of this note is to clear up these persistent
misconceptions that have appeared in the published literature.

\section{A qualitative argument}
\label{sec:qualitative.counterex}

As we show in this section with a qualitative argument, the problem
can be viewed as arising from the fact that in a manifold of dimension
$m>2$, the intersection of two metric spheres is not uniquely
specified by $m$ points. We demonstrate the issue in the context of
Delaunay balls. The problem is developed quantitatively in terms of
the Voronoi diagram in Section~\ref{sec:counter.ex}. 

We work exclusively on a three dimensional domain. The problem is a
local one, and we are not concerned with ``boundary conditions''; we
are looking at a coordinate patch on a densely sampled compact
$3$-manifold. 

One core ingredient in Delaunay's triangulation
result~\cite{delaunay1934} is that any triangle $\splxt$ is the face
of exactly two tetrahedra. This follows from the observation that a
triangle has a unique circumcircle, and that any circumscribing sphere
for $\splxt$ must include this circle. The affine hull of $\splxt$
cuts space into two components, and if $\splxt \in \del{\rem}{\mpts}$,
then it will have an empty circumscribing sphere centred at a point $c$ on the
line through the circumcentre and orthogonal to $\affhull{\splxt}$.
The point $c$ is contained on an interval on this line which contains
all the empty spheres for $\splxt$. The endpoints of the interval are
the circumcentres of the two tetrahedra that share $\splxt$ as a face.

The argument hinges on the assumption that the points are in general
position, and the uniqueness of the circumcircle for $\splxt$. If
there were a fourth vertex lying on that circumcircle, then there
would be three tetrahedra that have $\splxt$ as a face, but this
configuration would violate the assumption of general position.

Now if we allow the metric to deviate from the Euclidean one, no
matter how slightly, the guarantee of a well defined unique
circumcircle for $\splxt$ is lost. In particular, if three spheres $S_1$,
 $S_2$ and $S_3$ all circumscribe $\splxt$, their pairwise intersections
 will be different in general, i.e.,
 \begin{equation*}
   S_1 \cap S_3 \neq S_2 \cap S_3.
 \end{equation*}
 Although these intersections may be topological circles that are
 ``very close'' assuming the deviation of the metric from the
 Euclidean one is small enough, ``very close'' is not good enough when
 the only genericity assumption allows configurations that are
 arbitrarily bad. 

\begin{figure}
  \begin{center}
    \includegraphics[width=.8\columnwidth]{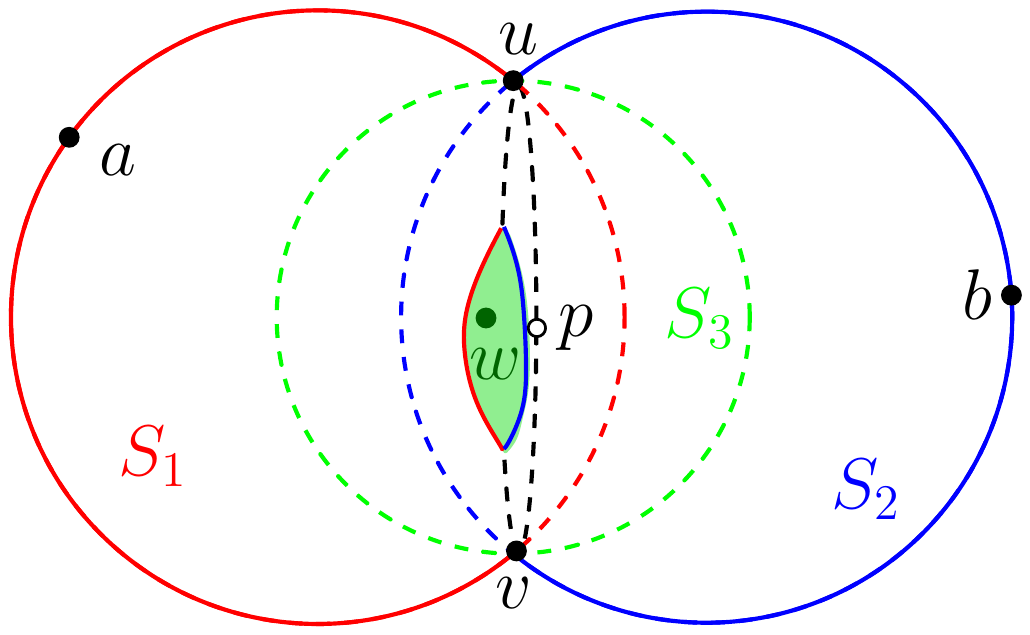}
  \end{center}
  \caption{In three dimensions, three closed geodesic balls can all
    touch three points, $u,v,p$, on their boundary and yet no one of
    them is contained in the union of the other two. 
In Euclidean space $\asimplex{u,v,p}$
    defines a unique circumcircle (shown dashed black) that is the
    intersection of any two distinct spheres that circumscribe
    $\asimplex{u,v,p}$, but if the metric deviates from Euclidean,
    this circle is no longer relevant; we can have 
three distinct spheres, $S_1$, $S_2$, and $S_3$ all circumscribing
$\asimplex{u,v,p}$ and such that $S_1 \cap S_3 \neq S_2 \cap S_3$.
This means that
    $\asimplex{u,v,p}$ can be a face of three Delaunay tetrahedra, e.g.,
    $S_1$, $S_2$ and $S_3$ could be circumscribing spheres for
    $\asimplex{a,u,v,p}$, $\asimplex{u,v,p,b}$, and
    $\asimplex{u,v,p,w}$, respectively.} 
  \label{fig:three.tets}
\end{figure}
The problem is illustrated in \Figref{fig:three.tets}, where
$\splxt = \asimplex{u,v,p}$. Here, circumspheres $S_1$ and $S_2$ of
the Delaunay tetrahedra $\asimplex{a,u,v,p}$ and $\asimplex{u,v,p,b}$
would contain any empty sphere $S_3$ that circumscribes $\splxt$ if
the metric were Euclidean, but any aberration in the metric may leave
a part of $S_3$ exposed to the outside. This means that in principle
another sample point $w$ could lie on~$S_3$, while $S_1$ and $S_2$
remain empty. Thus there would be a third Delaunay tetrahedron,
$\asimplex{u,v,p,w}$ that shares $\splxt$ as a face.

In dimension 2 this problem does not arise \cite{leibon1999,dyer2008sgp}. 
The essential difference between dimension 2 and the higher dimensions
can be observed by examining the topological intersection properties
of spheres. Specifically, two $(m-1)$-spheres intersect transversely
in an $(m-2)$-sphere. For a non-Euclidean metric, even if this
property holds for sufficently small geodesic spheres, it is only when
$m=2$ that the sphere of intersection of the Delaunay spheres of two
adjacent $m$-simplices is uniquely defined by the vertices of the
shared $(m-1)$-simplex.


%% file: scale_invariant.tex
%
%

\section{An obstruction to intrinsic Delaunay triangulations}
\label{sec:counter.ex}

We now explicitly show how density assumptions alone cannot escape
topological problems in the Delaunay complex.  The configuration
considered here may be recognised as similar to the one described
qualitatively in \Secref{sec:qualitative.counterex}, but here we
consider the Voronoi diagram rather than Delaunay balls.  As before we
work exclusively in a local coordinate patch on a densely sampled
compact $3$-manifold.

The idea is to consider four points in the $xz$-plane in Euclidean
space, and show that with a mild perturbation of the metric, this
tetrahedron exhibits exactly two distinct circumcentres. We construct
a perturbed metric such that for any sampling density, a tetrahedron
with two distinct Delaunay balls can exist. This means that the
Delaunay complex, defined as the nerve of the Voronoi diagram, will
not triangulate the manifold (this Delaunay tetrahedron will have two
triangle faces that are not the face of any other Delaunay
tetrahedron, as discussed below).

Delaunay's genericity condition~\cite{delaunay1934} naturally extends
to the setting where $\man$ is a manifold, and Leibon and
Letscher~\cite[p.~343]{leibon2000} explicitly assume genericity in
this sense:
\begin{definition}
  \label{def:ll.delaunay.gen}
  If $\man$ is a Riemannian $m$-manifold, the set
  $\mpts \subset \man$, is \defn{generic} if no subset of $m+2$ points
  lies on the boundary of a geodesic ball.
\end{definition}
Delaunay only imposed the constraint on empty balls, and he showed
that any (finite or periodic) point set in Euclidean space can be made
generic through an arbitrarily small affine perturbation. That a
similar construction of a perturbation can be made for points on a
compact Riemannian manifold has not been explicitly
demonstrated. However, in light of the construction we now present, it
seems that the question is moot when $m>2$, because an arbitrarily
small perturbation from degeneracy will not be sufficient to ensure
that a Delaunay triangulation exists.

For a point $x$ in a compact Riemannian manifold $\man$ (without
boundary), the injectivity radius at $x$ is the supremum of radii $r$
such that a geodesic ball with radius $r$, i.e.,
$\ballM{x}{r} = \{ y \in \man \mid \distM{x}{y} < r \}$, has the
property that each $y \in \ballM{x}{r}$ is connected to $x$ by a
unique minimising geodesic. The \defn{injectivity radius} of $\man$,
denoted $\injradM$, is the infimum of the injectivity radii at the
points of $\man$. When $\man$ is a compact manifold, $\injradM > 0$;
see Chavel~\cite[\S III.2]{chavel2006}.  

We say $\mpts \subset \man$ is \defn{$\samconst$-dense} if
$\distM{x}{\mpts} < \samconst$ for any $x \in \man$. If
$\distM{p}{q} \geq \tsamconst$ for all $p,q \in \mpts$, then $\mpts$
is \defn{$\tsamconst$-separated}. The set $\mpts$ is an
\defn{$\samconst$-net} if it is $\samconst$-dense and
$\samconst$-separated.

\begin{theorem}
  \label{thm:counter.ex}
  For every integer $m>2$ there exists a compact Riemannian
  $m$-manifold $\man$ such that for any $\samconst > 0$, with
  $\samconst < \injradM/2$, there exists a generic $\samconst$-net
  $\mpts \subset \man$ such that $\delMmpts$ does not triangulate
  $\man$. Furthermore, this holds true for all perturbed point sets
  $\mpts' \subset \man$ that are within a small positive Hausdorff
  distance of $\mpts$.
\end{theorem}

\begin{proof}
  Consider $3$-dimensional Euclidean space, parameterised by $x$, $y$
  and $z$, to be the parameter domain of a coordinate chart of a
  compact Riemannian $3$-manifold $\man$. Place points $u$ and $v$ at
  $\pm \samconst a$ on the $z$ axis, and points $w$ and $p$ at
  $\pm \samconst a (1 + \xi)$ on the $x$-axis. Here $\samconst > 0$ is
  the sampling parameter that parameterises the scale of the example,
  $a = \unfrac{1}{\sqrt{2}}$, and $\xi> 0$ is a small parameter that
  will be constrained by considerations below. We will fix a metric
  $\gdistM$ on $\man$, and the value of $\xi$ will depend on this
  metric, as well as the sampling parameter $\samconst$.

  With this configuration, $u$ and $v$ will share a Voronoi face in
  the Euclidean metric which will extend indefinitely in the
  $y$-direction, but have a very short extent in the $x$-direction,
  assuming $\xi$ is very
  small. \Figref{fig:face.views}\subref{sfig:xzvor} depicts the slice
  of this Voronoi diagram defined by the $xz$-plane, as viewed from
  somewhere on the negative $y$-axis.

  \begin{figure}[t]
    \begin{center}
      \subfloat[Euclidean metric]{
        \label{sfig:xzvor}
        \includegraphics[width=.45\columnwidth]{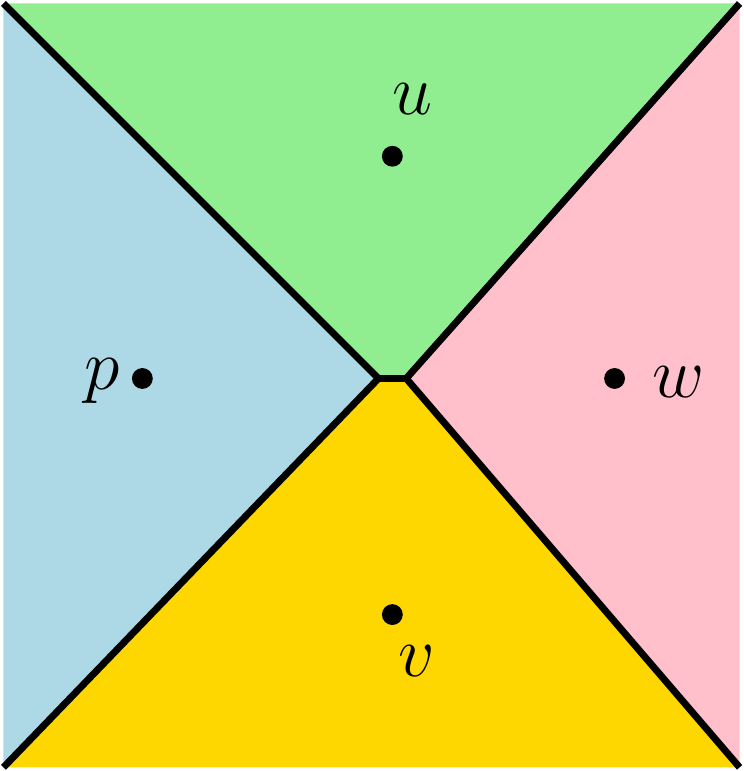} }
      \subfloat[Metric of $\man$]{
        \label{sfig:after.xzvor}
        \includegraphics[width=.45\columnwidth]{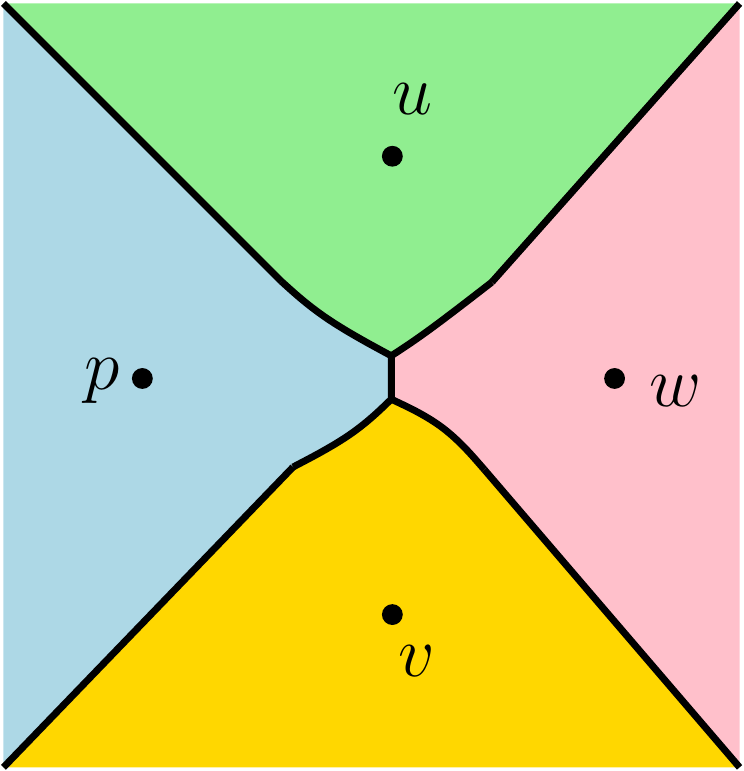} }
    \end{center}
    \caption{ A qualitative schematic of a vertical slice (the $y=0$
      plane) of the Voronoi diagram, seen from the negative $y$ axis. In
      this plane the Euclidean diagram \protect\subref{sfig:xzvor} is
      qualitatively different from the diagram
      \protect\subref{sfig:after.xzvor} defined by the metric $\gdistM$
      of $\man$.  }
    \label{fig:face.views}
  \end{figure}

  We now construct a metric on $\man$ such that the geodesic distance
  between $u$ and $v$ will be greater than the distance between $p$
  and $w$, thus qualitatively changing the structure of the Voronoi
  diagram in the $xz$-plane, as shown in
  \Figref{fig:face.views}\subref{sfig:after.xzvor}. To do this we
  employ a $\cinfty$ non-negative bump function
  $f\colon \reel \to \reel_{\geq 0}$ that is symmetric:
  $f(-y) = f(y)$, and satisfies $f'(y) \leq 0$ for $y>0$ within our
  region of interest, and a strict inequality $f'(y)<0$ for $y>0$ in
  some open neighbourhood of the origin.  Thus there is a unique
  maximum at the origin.  As an explicit example, we may assume that
  the underlying set of $\man$ is the torus $\reel^3/\ints^3$, and
  that $f(y) = A(1 + \cos(\pi y))$, where the amplitude is constrained
  by $A \leq 3/8$, as described below.
  We define the metric tensor of $\man$ in our coordinate system by
  \begin{equation*}
    g(q) =
    \begin{pmatrix}
      1 & 0 & 0 \\
      0 & 1 & 0 \\
      0 & 0 & 1 + f(y(q))
    \end{pmatrix},
  \end{equation*}
  where $y(q)$ denotes the $y$-coordinate of the point $q$.

  By integrating along the $z$-axis we obtain an upper bound on the
  geodesic distance between $v$ and $u$:
  \begin{equation}
    \label{eq:gdistuv}
    \distM{v}{u} \leq 2\int_0^{a\samconst}  (1 + f(0))^{1/2}dz
    = 2a \samconst (1 + f(0))^{1/2}.
  \end{equation}
  In fact, if the right hand side of \eqref{eq:gdistuv} is less than
  the injectivity radius of $\man$, then equality with $\distM{v}{u}$
  is attained. Indeed, in this case the curve represented by the
  $z$-axis is the minimising geodesic between $v$ and $u$. This
  follows from the symmetry of the metric: the reflection through the
  $xz$-plane $(x,y,z) \mapsto (x, -y, z)$ is an isometric involution
  that leaves points in the $xz$-plane fixed. Since an isometry maps
  geodesics to geodesics, and the minimising geodesic between $v$ and
  $u$ is unique, it must lie in the $xz$-plane, and, again by
  symmetry, the minimising curve is the straight line in the parameter
  domain.

  We require that all four points lie within a ball centred at the
  origin, and with diameter less than $\injradM$.  The final
  constraint that we impose on $f$ is that $f(0) \leq 3/4$. Another
  reason for this constraint is described below, but in particular it
  implies that $a (1 + f(0))^{1/2} < 1$. Then we see that, provided
  that $\samconst < \injradM/2$,
  \begin{equation}
    \label{eq:distuv}
    \distM{v}{u} = 2a \samconst (1 + f(0))^{1/2} <\injradM.
  \end{equation}

  We also require that $\distM{p}{w} < \distM{u}{v}$, so that we have
  the qualitative Voronoi diagram in the $xz$-plane depicted by
  \Figref{fig:face.views}\subref{sfig:after.xzvor}. Observing that
  \begin{equation*}
    \distM{p}{w} = \distE{p}{w} = 2a\samconst(1+\xi), 
  \end{equation*}
  and equating with \eqref{eq:distuv}, we see that
  $\distM{p}{w} = \distM{u}{v}$ if $\xi$ is equal to the critical
  value defined by
  \begin{equation*}
    \xi_0 = (1+ f(0))^{1/2} - 1,
  \end{equation*}
  and so we require $\xi < \xi_0$.

  The goal is to show that the tetrahedron $\sigma = \{u,v,w,p \}$ can
  be a Delaunay tetrahedron admitting two distinct Delaunay balls in
  an $\samconst$-net $P$.  To that end, we require that it has a
  circumcentre $c$ on the positive $y$-axis, at a distance less than
  $\samconst$ from the vertices ($-c$ will also be a
  circumcentre). Let the $y$-coordinate of $c$ be $b\samconst$. We
  argue that for any $b>0$, there is a $\xi<\xi_0$ such that
  $c=(0,b\samconst,0)$ and $-c$ are circumcentres of $\sigma$.

  The geodesic distance between $u$ and a given $c$ is no greater than
  the geodesic length of the straight line between them in the
  parameter domain:
  \begin{equation}
    \label{eq:upper.bnd.xi}
    \begin{split}
      \distM{u}{c} 
      &\leq \int_0^{b\samconst} \left( (1+ f(y))
        \frac{a^2}{b^2} + 1 \right)^{\frac{1}{2}} dy \\
      &< \int_0^{b\samconst} \left( (1+ f(0)) \frac{a^2}{b^2} + 1
      \right)^{\frac{1}{2}} dy = b\samconst
      \left((1+\xi_0)^2\frac{a^2}{b^2} +1 \right)^{\frac12}.
    \end{split}
  \end{equation}
  Also, the geodesic distance between $u$ and $c$ is greater than the
  Euclidean distance:
  \begin{equation*}
    \distM{u}{c} > \distE{u}{c} 
    = b\samconst\left(\frac{a^2}{b^2}+1 \right)^{\frac12}.
  \end{equation*}
  It follows then that for any $b>0$, there exists a positive
  $\txi(b) <\xi_0$ such that
  \begin{equation}
    \label{eq:def.xi}
    \distM{u}{c} 
    = b\samconst \left( (1+\txi(b))^2 \frac{a^2}{b^2} + 1 \right)^{\frac{1}{2}}.
  \end{equation}
  Since $\distM{u}{c}$ varies smoothly with $b$, $\txi$ must be a
  smooth function.

  The symmetry of the metric implies that the unique minimizing
  geodesic between $p$ and $c$ must lie in the $xy$-plane, where the
  metric coincides with the Euclidean one. It follows that, for
  $p=(-(1+\xi)a\samconst,0,0)$, we have
  \begin{equation}
    \label{eq:distpc}
    \distM{p}{c} = \distE{p}{c} 
    = 
    b\samconst \left( (1+\xi)^2 \frac{a^2}{b^2} + 1 \right)^{\frac{1}{2}}.
  \end{equation}
  Thus when $\txi(b)=\xi$ we have $\distM{u}{c}=\distM{p}{c}$, and it
  follows (by symmetry) that $c$ is a circumcentre for $\sigma$.

  We need to keep $b$ small enough that $\distM{u}{c} < \samconst$, so
  that $\sigma$ can be a Delaunay tetrahedron in an
  $\samconst$-net. It is sufficient to demand that
  $b \leq \bmax = a/2$; then the fact that $\distM{u}{c} < \samconst$
  follows from the constraint $f(0) \leq 3/4$ and the definitions
  $a^2 = 1/2$, and $b^2 \leq \bmax^2 = 1/8$:
  \begin{equation*}
    \distM{u}{c}
    < b\samconst
    \left( (1+ f(0)) \frac{a^2}{b^2} + 1 \right)^{\frac{1}{2}}
    \leq \samconst.
  \end{equation*}
  Also, it follows from the construction that the vertices of $\sigma$
  meet the separation criterion of an $\samconst$-net: geodesic
  distances between the points are bounded below by the Euclidean
  distances in the parameter domain, e.g.,
  $\distM{u}{w} > \norm{u-w} = a\samconst(1+(1+\xi)^2)^{1/2} >
  \samconst$.
  Thus, $\sigma$ can be realised as a Delaunay simplex with exactly
  two distinct Delaunay balls in an $\samconst$-net.

  The configuration can be realised as part of an $\samconst$-net so
  that the two circumcentres of $\sigma$ remain Voronoi
  vertices. Furthermore, this is not a degenerate configuration in the
  sense of \Defref{def:ll.delaunay.gen}, because the requirements of
  an $\samconst$-net are met without placing any further vertices on
  the Delaunay spheres associated to $\sigma$; i.e., $\sigma$ is not
  the proper face of any Delaunay simplex.

  When we deviate from a Euclidean metric we can encounter
  degeneracies different from those described by
  \Defref{def:ll.delaunay.gen}. For example, if we were to choose
  $\xi=\xi_0$ in our example, then $\sigma$ would have a unique
  circumcentre (at the origin), and it could be a Delaunay simplex
  that is not the proper face of any other. But an arbitrarily small
  perturbation, sending $u$ or $v$ towards the origin, for example,
  would mean that $\sigma$ has no circumcentre at all.

  However, at least for some positive $b\leq\bmax$, the presented
  example is not degenerate in any strict sense: there is some number
  $\rho>0$ such that the vertices of $\sigma$ can be each
  independently displaced by a distance $\rho$ without disturbing the
  combinatorial structure of the Voronoi diagram.  In other words,
  this bad configuration is represented by a set of positive measure
  in the configuration space.

  In order to demonstrate this it is sufficient to show that the
  circumcentres of $\sigma$ depend continuously on the position of the
  vertices, at least in a neighbourhood of $s_0 = (p,v,w,u)$. This can
  be shown by demonstrating that the origin is a regular value of the
  map
\begin{equation}
  \begin{split}
    h\colon \R^3 &\to \R^3 \\
    q &\mapsto \bigl(\distM{q}{u} - \distM{q}{p}, \distM{q}{v} -
    \distM{q}{p}, \distM{q}{w} - \distM{q}{p}\bigr).
  \end{split}
  \label{eq:cc.equation}
\end{equation}
Indeed, if this is the case, then the implicit function theorem
applied to the function $F\colon\R^{12}\times\nobreak \R^3\to\R^3$
defined by \eqref{eq:cc.equation}, such that $F(s_0,q) = h(q)$,
ensures that for all vertex positions $s$ in a neighbourhood of $s_0$,
there will be circumcentres (solutions to $F(s,q)=0$) close to $c$ and
$-c$.

In order to estimate the Jacobian determinant of $h$ at $q=c$, let
$H(q) = H(x,y,z) = \distM{q}{u} - \distM{q}{p}$ be the first component
of $h$, and observe that, by symmetry,
$H(x,y,-z) = \distM{q}{v} - \distM{q}{p}$ is the second component. For
the third component, $W(q) = \distM{q}{w}-\distM{q}{p}$, we observe
that $W(q)=0$ if $x(q)=0$. Therefore, denoting partial derivatives
with subscripts, we have $W_y = W_z = 0$ at $q=c$. Also, we observe
directly that $W_x$ is negative. Let $W_x(c) = t < 0$. Therefore the
Jacobian matrix has the form
\begin{equation*}
  Dh(c) =
  \begin{pmatrix}
    H_x & H_y & H_z \\
    H_x & H_y & -H_z \\
    t & 0 & 0
  \end{pmatrix},
\end{equation*}
and to show that its determinant is nonzero, we need to verify that
$H_y$ and $H_z$ are nonzero at~$c$.

By inspection $H_z(c) < 0$ since $c$ is a critical point of
$\distM{q}{p}$ with respect to variation in the $z$ direction, and
$\distM{q}{u}$ is strictly decreasing at $c$ as $q$ moves up a
vertical line. For $H_y(c)$, we exploit the fact that we have explicit
expressions, \eqref{eq:def.xi} and \eqref{eq:distpc}, that describe
$\distM{u}{c}$ and $\distM{p}{c}$. When differentiating with respect
to $b$ at the point where $\txi(b)=\xi$, we find
\begin{equation}
\label{eq:expr.Hy}
  H_y(c) = \partial_b\distM{c}{u} - \partial_b\distM{c}{p}
= \frac{a^2(1+\xi)\txi'\samconst}{((1+\xi)^2a^2+b^2)^{\frac12}}.
\end{equation}
\Eqnref{eq:upper.bnd.xi} shows that for any $b>0$, we have
$\txi(b)<\xi_0 = \txi(0)$, and in particular, this is true for
$b=\bmax$. So it follows from the mean value theorem that there is a
$b \in (0,\bmax]$ such that $\txi'(b)<0$, and \eqref{eq:expr.Hy}
implies that when $\xi$ is chosen so that this $b$ represents a
cicumcentre, we have $H_y(c) < 0$. Thus this configuration is stable
with respect to small perturbations of the vertices.

\begin{figure}[t]
  \begin{center}
    \subfloat[$xy$-plane seen from positive $z$-axis]{
      \label{sfig:from.above}
      \includegraphics[width=.8\columnwidth]{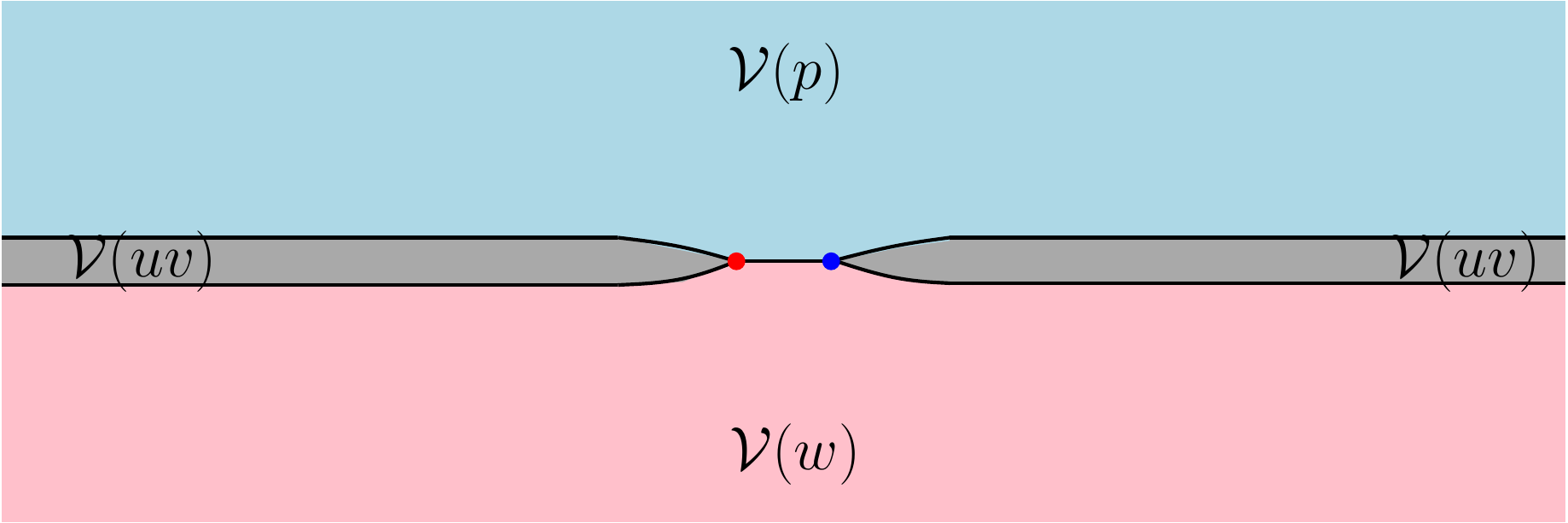} } \\
    \subfloat[$yz$-plane seen from positive $x$-axis]{
      \label{sfig:front.view}
      \includegraphics[width=.8\columnwidth]{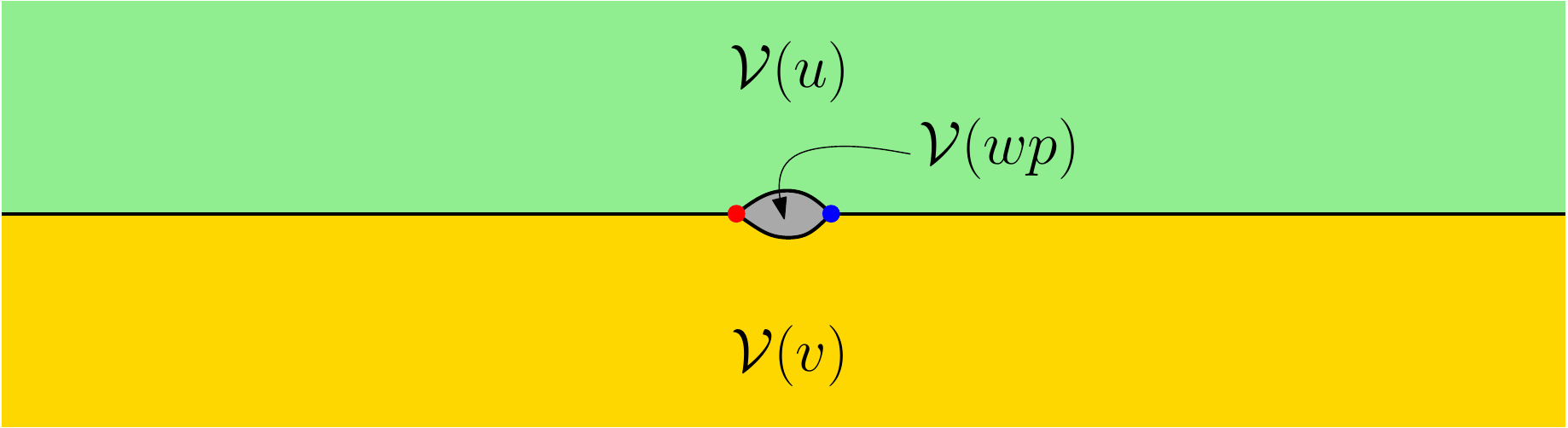} }
  \end{center}
  \caption{Looking at cross-sections; the positive $y$-direction is to
    the right. The tetrahedron $\sigma = {p,u,v,w}$, admits exactly
    two small circumballs with distinct centres (the red and blue
    points).}
  \label{fig:side.views}
\end{figure} 
This constructed Delaunay tetrahedron $\sigma$ with two distinct
circumcentres implies a defect in the Voronoi diagram that prevents
$\delMmpts$ from triangulating $\man$.  A schematic view of the
Voronoi diagram associated with $\sigma$ is depicted in
\Figref{fig:side.views}.  The circumumcentres of $\sigma$ are the
Voronoi vertices shown in the diagram. A Voronoi edge corresponds to a
Delaunay triangle, and the cofaces of the triangle are the Delaunay
tetrahedra that correspond to the Voronoi vertices at the ends of the
Voronoi edge. So for triangle $t=\{u,p,w\}$, since the Voronoi vertex
at each end of $\vorcellman{t}$ represents the same simplex, $\sigma$,
there will be only one coface of $t$ (likewise for triangle
$\{v,p,w\}$). Since our manifold $\man$ has no boundary, $\delMmpts$
cannot admit a triangulation of~$\man$.

Finally, observe that this kind of counter-example to Delaunay
triangulations will also exist in higher dimensional manifolds. For
example, the same basic construction when $m>3$ can be obtained by
starting with a regular simplex $\sigma$ of dimension $(m-2)$ such
that its circumcentre is at the origin, and its circumradius is
$\samconst a(1+\xi)$, with $a=1/\sqrt{2}$, and such that
$\affhull{\sigma}=\{x\in\R^m \mid x_{m-1}=x_{m}=0\}$. Then place two
additional vertices at $\pm a\samconst$ on the $x_{m-1}$-axis to
obtain an $m$-simplex.  The metric $g$ is then the diagonal matrix
with $g_{ii}(q)= (1+f(x_{m-1}(q)))$ if $i=m$ and $g_{ii}(q)=1$
otherwise, and $f$ defined exactly as before.
\end{proof}

The qualitative construction of the previous section
(\Figref{fig:three.tets}) is similar to the one we have considered
here except that it conceptually employs a different perturbation
function, whose support does not include the $y$-axis, so that the
simplex $\{u,v,p,w\}$ can be dual to a single stable Voronoi
vertex. The Delaunay perspective on the example represented in
\Figref{fig:side.views} would be similar to \Figref{fig:three.tets},
except that there would be two distinct green Delaunay balls for
$\sigma$.


%% file: discussion.tex
%

\section{Discussion}

We have shown that for constructing a Delaunay triangulation of an
arbitrary compact Riemannian manifold, a sampling density requirement
is not sufficient in general. One approach to this problem might be to
constrain the kind of metrics that are considered. However, the
example shown here exhibits a problem even with a mild deviation from
homogeneity, so the admissible metrics will be severely limited. We
expect this kind of problem to arise any time one attempts to contruct
an anisotropic triangulation using the empty sphere property in
dimension $m>2$, unless the $(m-1)$-spheres circumscribing $m$ points
are constrained to all intersect in a unique $(m-2)$-sphere defined by
those points. In particular, this obstruction will be encountered when
attempting to extend the planar anisotropic triangulation result of
Canas and Gortler~\cite{canas2012} to higher dimensions.

An alternate approach is to constrain the kinds of point sets
considered, so that the problematic configurations do not arise. We
have developed this approach in other
work~\cite{boissonnat2013manmesh.arxiv}.

The example demonstrated here shows that even asymptotic arguments
involving Delaunay triangulations on manifolds must be handled with
care.  Clarkson~\cite{clarkson2006} remarked that an implication of
Leibon and Letscher's claim \cite{leibon2000} is that for $m+1$
sufficiently close points there is a unique circumsphere with small
radius. Our construction is an explicit counter-example.

It is worth emphasising that the problems discussed here only arise
when the dimension is greater than $2$. Density based sampling
criteria for Delaunay triangulation of two dimensional manifolds has
been validated \cite{leibon1999,dyer2008sgp}.
